\documentclass[11pt]{article}
\usepackage{epic,latexsym,amssymb}
\usepackage{color}
\usepackage{tikz}
\usepackage{amsfonts,epsf,amsmath}
\usepackage{lineno}

\usepackage{subfig}
\usepackage{epic}
\usepackage{amssymb}
\usepackage{latexsym}
\usepackage{tikz}
\usetikzlibrary{decorations.markings}
\usepackage{pgf}

\usetikzlibrary{arrows}

\usepackage{comment}

\newtheorem{theorem}{Theorem}

\newtheorem{definition}{Definition}
\newtheorem{claim}{Claim}

\newtheorem{corollary}{Corollary}

\newtheorem{proposition}{Proposition}

\newtheorem{problem}{Problem}

\newcommand{\CG}{{\rm CG}}
\newcommand{\CP}{{\rm CP}}
\newcommand{\NC}{{\rm NC}}

\newcommand{\QED}{$\Box$}

\newcommand{\smallqed}{{\tiny ($\Box$)}}

\newcommand{\1}{\vspace{0.1cm}}

\newcommand{\proof}{\noindent\textbf{Proof. }}

\let\oldenumerate\enumerate
\renewcommand{\enumerate}{
  \oldenumerate
  \setlength{\itemsep}{1pt}
  \setlength{\parskip}{0pt}
  \setlength{\parsep}{0pt}
}

\def\vertex(#1){\put(#1){\circle*{2}}}
\def\vertexo(#1){\put(#1){\circle{2}}}
\def\vert(#1){\put(#1){\circle*{1.5}}}
\def\verto(#1){\put(#1){\circle{1.5}}}
\def\lab(#1)#2{\put(#1){\makebox(0,0)[c]{#2}}}
\begin{document}

\title{On the coalition number of trees}

\author{$^1$Davood Bakhshesh, \, $^2$Michael A. Henning\thanks{Research supported in part by the University of Johannesburg.}, \, and $^3$Dinabandhu Pradhan\thanks{Corresponding author.} \\ \\
$^1$Department of Computer Science\\
University of Bojnord \\
Bojnord, Iran \\
\small \tt Email: d.bakhshesh@ub.ac.ir \\
\\
$^2$Department of Mathematics and Applied Mathematics \\
University of Johannesburg \\
Auckland Park, 2006 South Africa\\
\small \tt Email: mahenning@uj.ac.za\\
\\
$^3$Department of Mathematics and Computing \\
Indian Institute of Technology (ISM) \\
Dhanbad, India \\
\small \tt Email: dina@iitism.ac.in}

\date{}
\maketitle

\begin{abstract}
Let $G$ be a graph with vertex set $V$ and of order~$n = |V|$, and let $\delta(G)$ and $\Delta(G)$ be the minimum and maximum degree of $G$, respectively. Two disjoint sets $V_1, V_2 \subseteq V$ form a coalition in $G$ if none of them is a dominating set of $G$ but their union $V_1\cup V_2$ is.  A vertex partition $\Psi=\{V_1,\ldots, V_k\}$ of  $V$ is a coalition partition of $G$ if  every set $V_i\in \Psi$ is either a dominating set of $G$ with the cardinality $|V_i|=1$,  or is not a dominating set  but for some $V_j\in \Psi$,  $V_i$ and $V_j$ form a coalition.  The maximum cardinality of a  coalition partition of $G$ is the coalition number $\mathcal{C}(G)$ of $G$. Given a coalition partition $\Psi = \{V_1, \ldots, V_k\}$ of $G$, a coalition graph $\CG(G, \Psi)$ is associated on $\Psi$ such that there is a one-to-one correspondence between its vertices and the members of $\Psi$, where two vertices of $\CG(G, \Psi)$ are adjacent if and only if the corresponding sets form a coalition in $G$. In this paper, we partially solve one of the open problems posed in Haynes et al.~\cite{coal0} and we solve two open problems posed by Haynes et al.~\cite{coal1}. We characterize all graphs $G$ with $\delta(G) \le 1$ and $\mathcal{C}(G)=n$, and we characterize all trees $T$ with $\mathcal{C}(T)=n-1$. We determine the number of coalition graphs that can be defined by all coalition partitions of a given path.  Furthermore, we show that there is no universal coalition path, a path whose coalition partitions defines all possible coalition graphs.
\end{abstract}

\indent
{\small \textbf{Keywords:}  Coalition number; Domination number; Coalition partition; Coalition graphs.} \\
\indent {\small \textbf{AMS subject classification:} 05C69}

\section{Introduction}

Let $G$ be a graph with vertex set $V = V(G)$. Throughout this paper, we only consider graphs without multiple edges and loops. Two vertices are \emph{neighbors} if they are adjacent. A \emph{dominating set} of a graph $G$ is a set $S$ of vertices of $G$ such that every vertex in $G$ is dominated by a vertex in $S$, where a vertex \emph{dominates} itself and its neighbors. Equivalently, $S \subseteq V$ is a dominating set of $G$ if every vertex of $V \setminus S$ is adjacent to a vertex of $S$. The minimum cardinality of a dominating set of $G$ is the \emph{domination number} of $G$, denoted by $\gamma(G)$. If $X,Y \subseteq V(G)$, then the set $X$ \emph{dominates} the set $Y$ if every vertex in $Y$ is dominated by at least one vertex in $X$. We refer the reader to the books~\cite{HaHeHe-20,HaHeHe-21} to study an overview of dominating sets in graphs.

In~2020, Haynes, Hedetniemi, Hedetniemi, McRae, and Mohan~\cite{coal0} presented a graph theoretic model of a coalition, and introduced the concept of a \emph{coalition} in graphs. They defined a pair of sets $V_1, V_2 \subseteq V$ to be a \emph{coalition} in $G$ if none of them is a dominating set of $G$ but $V_1 \cup V_2$ is. Such a pair $V_1$ and $V_2$ is said to \emph{form a coalition}, and are called \emph{coalition partners}.  A vertex partition $\Psi = \{V_1,\ldots, V_k\}$ of  $V$ is a \emph{coalition partition} of $G$, abbreviated a $c$-\emph{partition} in~\cite{coal0}, if every set~$V_i\in \Psi$ is either a dominating set of $G$ with cardinality $|V_i|=1$,  or is not a dominating set  but for some $V_j\in \Psi$,   $V_i$ and $V_j$ form a coalition.  The maximum cardinality of a coalition partition of $G$ is called the \emph{coalition number} of $G$,  denoted by $\mathcal{C}(G)$. A coalition partition of $G$ of cardinality  $\mathcal{C}(G)$ is called a \emph{$\mathcal{C}$-partition of $G$}. A motivation of this graph theory model of a coalition is given by Haynes et al. in their series of papers on coalitions in~\cite{coal0,coal1,coal2,coal3}.

Given a coalition partition $\Psi = \{V_1, \ldots, V_k\}$ of $G$, a coalition graph $\CG(G, \Psi)$ is associated on $\Psi$ such that there is a one-to-one correspondence between its vertices and the members of $\Psi$, where two vertices of $\CG(G, \Psi)$ are adjacent if and only if the corresponding sets form a coalition in $G$.

For notation and graph theory terminology not defined herein, we in general follow~\cite{HeYe-book}. Specifically, let $G$ be a graph with vertex set $V(G)$ and edge set $E(G)$, and of order~$n(G) = |V(G)|$ and size $m(G) = |E(G)|$. For a set of vertices $S\subseteq V(G)$, the subgraph induced by $S$ is denoted by $G[S]$. Two vertices in $G$ are \emph{neighbors} if they are adjacent. The \emph{open neighborhood} $N_G(v)$ of a vertex $v$ in $G$ is the set of neighbors of $v$, while the \emph{closed neighborhood} of $v$ is the set $N_G[v] = \{v\} \cup N(v)$. We denote the \emph{degree} of $v$ in $G$ by $\deg_G(v) = |N_G(v)|$. The minimum and maximum degree in $G$ is denoted by $\delta(G)$ and $\Delta(G)$, respectively. An \emph{isolated vertex} is a vertex of degree~$0$, and an \emph{isolate}-\emph{free graph} is a graph that contains no isolated vertex. A vertex of degree~$1$ is called a \emph{leaf}, and its unique neighbor a \emph{support vertex}. A graph is \emph{isolate}-\emph{free} if it contains no isolated vertex. A vertex of degree~$n(G)-1$ in $G$ is a \emph{universal vertex}, also called a \emph{full vertex} in the literature, of $G$. For a set $S \subseteq V(G)$, its \emph{open neighborhood} is the set $N_G(S) = \cup_{v \in S} N_G(v)$, and its \emph{closed neighborhood} is the set $N_G[S] = N_G(S) \cup S$. If the graph $G$ is clear from the context, we omit writing it in the above expressions. For example, we simply write $V$, $E$, $n$, $m$, $N(v)$ and $N(S)$ rather than $V(G)$, $E(G)$, $n(G)$, $m(G)$, $N_G(v)$ and $N_G(S)$, respectively.

A \emph{vertex cover} of a graph $G$ is a set $S$ of vertices such that every edge in $E(G)$ is incident with at least one vertex in~$S$. The \emph{vertex covering number} $\beta(G)$, also denoted $\tau(G)$ in the literature, is the minimum cardinality of a vertex cover of $G$. For a positive integer $k$, we let $[k] = \{1, \ldots, k\}$.
	
\subsection{Motivation and known results}

In this paper, we continue the study of coalitions in graphs. Our immediate aim is to answer, in part or fully, the following intriguing problems posed by the Haynes et al. in their recent series of paper given in~\cite{coal0,coal1,coal2}. Haynes et al.~\cite{coal0} posed the following open problem.
	
\begin{problem}{\rm (\cite{coal0})}
\label{prob1}
{\rm Characterize the graphs $G$ satisfying $\mathcal{C}(G)=n(G)$.}
\end{problem}

For $r, s \ge 1$, a \emph{double star} $S(r,s)$ is a tree with exactly two (adjacent) vertices that are not leaves, with one of the vertices having $r$ leaf neighbors and the other $s$ leaf neighbors. The double star $S(2,2)$, for example, is illustrated in Figure~\ref{f:fig1}(a).  The bull graph $B$, illustrated in Figure~\ref{f:fig1}(b), is a graph obtained from a triangle by adding two disjoint pendant edges. Let $F_1$ be obtained from a bull graph by deleting one of the vertices of degree~$1$, and let $F_2$ be obtained from a $4$-cycle by adding a pendant edges. The graphs $F_1$ and $F_2$ are illustrated in Figures~\ref{f:fig1}(c) and~\ref{f:fig1}(d), respectively.

\begin{figure}[htb]
\begin{center}
\begin{tikzpicture}[scale=.8,style=thick,x=1cm,y=1cm]
\def\vr{2.75pt} 
\path (0,0) coordinate (v1);
\path (0,2) coordinate (v3);
\path (1.25,1) coordinate (v4);
\path (2.25,1) coordinate (v5);
\path (3.5,0) coordinate (v6);
\path (3.5,2) coordinate (v7);
\draw (v1)--(v4);
\draw (v3)--(v4)--(v5)--(v6);
\draw (v5)--(v7);
\draw (v1) [fill=white] circle (\vr);
\draw (v3) [fill=white] circle (\vr);
\draw (v4) [fill=white] circle (\vr);
\draw (v5) [fill=white] circle (\vr);
\draw (v6) [fill=white] circle (\vr);
\draw (v7) [fill=white] circle (\vr);
\draw (1.75,-0.75) node {{\small (a) $S(2,2)$}};
\path (6.5,0) coordinate (u1);
\path (5.5,1) coordinate (u2);
\path (5.5,2) coordinate (u3);
\path (7.5,1) coordinate (u4);
\path (7.5,2) coordinate (u5);
\draw (u1)--(u2)--(u4)--(u1);
\draw (u2)--(u3);
\draw (u4)--(u5);
\draw (u1) [fill=white] circle (\vr);
\draw (u2) [fill=white] circle (\vr);
\draw (u3) [fill=white] circle (\vr);
\draw (u4) [fill=white] circle (\vr);
\draw (u5) [fill=white] circle (\vr);
\draw (6.5,-0.75) node {{\small (b) $B_1$}};
\path (10.5,0) coordinate (w1);
\path (9.5,1) coordinate (w2);
\path (9.5,2) coordinate (w3);
\path (11.5,1) coordinate (w4);
%
\draw (w1)--(w2)--(w4)--(w1);
\draw (w2)--(w3);
%
\draw (w1) [fill=white] circle (\vr);
\draw (w2) [fill=white] circle (\vr);
\draw (w3) [fill=white] circle (\vr);
\draw (w4) [fill=white] circle (\vr);
\draw (10.5,-0.75) node {{\small (c) $F_1$}};
\path (13.5,0) coordinate (x1);
\path (13.5,1) coordinate (x2);
\path (13.5,2) coordinate (x3);
\path (15.5,1) coordinate (x4);
\path (15.5,0) coordinate (x5);
\draw (x1)--(x2)--(x4)--(x5)--(x1);
\draw (x2)--(x3);
\draw (x1) [fill=white] circle (\vr);
\draw (x2) [fill=white] circle (\vr);
\draw (x3) [fill=white] circle (\vr);
\draw (x4) [fill=white] circle (\vr);
\draw (x5) [fill=white] circle (\vr);
\draw (14.5,-0.75) node {{\small (d) $F_2$}};
\end{tikzpicture}
\end{center}
\begin{center}
\vskip -0.5 cm
\caption{Four graphs of small orders}
\label{f:fig1}
\end{center}
\end{figure}
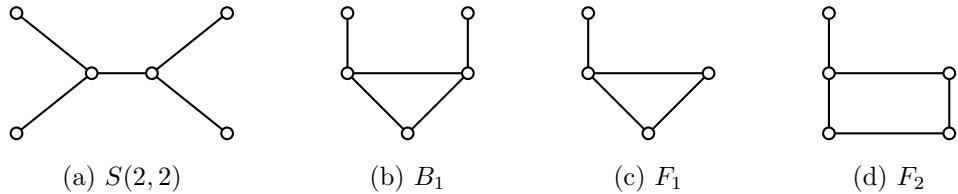
	
Let $\cal F$  be a family of graphs consisting of 18 graphs: $K_1$, $K_2$, $\overline{K_2}$, $K_1\cup K_2$, $P_3$, $K_3$, $K_{1,3}$, $2K_2$, $P_4$,  $C_4$, $F_1$,  $K_4-e$, $P_2\cup P_3$, $F_2$, $B_1$, $P_5$, $S(1,2)$, and $S(2,2)$. Haynes et al.~\cite{coal1} proved the following result.

\begin{theorem}{\rm (\cite{coal1})}
\label{t:known1}
If $\Psi$ is a coalition partition of a path $P_k$, then $\CG(P_k,\Psi) \in {\cal F}$.
\end{theorem}

A path is called a \emph{universal coalition path} if all 18 graphs of $\cal F$ can be defined by the coalition partitions of the path. Haynes et al.~\cite{coal1} posed the following two open problems.
	
\begin{problem}{\rm (\cite{coal1})}
\label{prob2}
{\rm Given a positive integer $k$, how many coalition graphs can be defined by the coalition partitions of a path $P_k$?}
\end{problem}
\begin{problem}{\rm (\cite{coal1})}
\label{prob3}
{\rm	Does there exist a positive integer $k$ such that all 18 graphs of $\cal F$ can be defined by the coalition partitions  of $P_k$? If so, what is the smallest universal coalition path?}
\end{problem}

In this paper, we  characterize all graphs $G$ with $\delta(G)=1$ and $\mathcal{C}(G)=n$. Moreover, we characterize all trees $T$ with  $\mathcal{C}(T)=n$  and all trees $T$ with $\mathcal{C}(T)=n-1$. This solves part of Problem~\ref{prob1}. On the other hand, we solve Problem~\ref{prob2} and Problem~\ref{prob3}. In particular, we theoretically and empirically determine the number of coalition graphs that can be defined by the coalition partitions of path $P_k$. Consequently, we show that there is no universal coalition path.

The coalition number of a path and a cycle is determined in~\cite{coal0}.

\begin{theorem}{\rm (\cite{coal0})}
\label{thm:path}
The following hold for a path $P_n$ and a cycle $C_n$. \\ [-24pt]
\begin{enumerate}
\item[{\rm (a)}] $C(P_n) = n$ if $n \le 4$, $C(P_n) = 4$ if $n = 5$, $C(P_n) = 5$ if $6 \le n \le 9$, and $C(P_n) = 6$ if $n \ge 10$.
\item[{\rm (b)}] $C(C_n) = n$ if $3 \le n \le 6$, $C(C_n) = 5$ if $n = 7$, and $C(C_n) = 6$ if $n \ge 8$.
\end{enumerate}
\end{theorem}

As consequence of Theorem~\ref{thm:path}, we have the following result.

\begin{corollary}{\rm (\cite{coal0})}
\label{cor:path}
If $G$ is a path or a cycle, then $C(G) \le 6$.
\end{corollary}

We shall need the following properties of a $c$-partition in a graph.

\begin{proposition}{\rm (\cite{coal0})}
\label{prop:Delta}
If $\Psi$ is a $c$-partition of a graph $G$, then $\Delta(\CG(G,\Psi)) \le \Delta(G) + 1$.
\end{proposition}

By Proposition~\ref{prop:Delta}, if $\Psi$ is a $c$-partition of a graph $G$ and $S\in \Psi$, then $S$ is in at most $\Delta(G)+1$ coalitions in $\Psi$. The following result is stated without proof in~\cite{coal2}. For completeness we present a proof of this elementary property of a $c$-partition in a graph.

\begin{proposition}{\rm (\cite{coal2})}
\label{p:path}
If $\Psi$ is a $c$-partition of a graph $G$ and $v$ is an arbitrary vertex of $G$, then $\beta(\CG(G,\Psi)) \le \deg_G(v) + 1$.
\end{proposition}
\begin{proof} Let $G$ with a graph with $c$-partition $\Psi = \{V_1, V_2, \ldots, V_k\}$, and let $v \in V(G)$. Let $\Psi'$ be a subset of $\Psi$, such that $V_i \in \Psi'$ if and only if $V_i$ contains a vertex of $N[v]$.  Hence, $|\Psi'| \le |N[v]| = \deg_G(v) + 1$.   Now consider an arbitrary edge $V_iV_j$ in $\CG(G,\Psi)$. Since the set $V_i \cup V_j$ is a dominating set of $G$, at least one of $V_i$ and $V_j$ contains a vertex from $N[v]$ in $G$.  Thus, at least one of $V_i$ and $V_j$ is in $\Psi'$. Hence, $\Psi'$ is a vertex cover of $\CG(G,\Psi)$, and so $\beta(\CG(G,\Psi)) \le |\Psi'| \le \deg_G(v) + 1$.~\QED
\end{proof}

\medskip
The following upper bounds on the coalition number of a graph are established in~\cite{coal2}.

\begin{theorem}{\rm (\cite{coal2})}
\label{thm:bound1}
If $G$ is a graph with $\delta = \delta(G)$ and $\Delta = \Delta(G)$, then the following hold. \\ [-24pt]
\begin{enumerate}
\item[{\rm (a)}] $C(G) \le \frac{1}{4}( \Delta + 3)^2$. \1
\item[{\rm (b)}] If $\delta < \frac{1}{2}\Delta$, then $C(G) \le (\delta + 1)(\Delta - \delta + 2)$.
\end{enumerate}
\end{theorem}

As a consequence of Corollary~\ref{cor:path} and Theorem~\ref{thm:bound1}(b), we have the following result.

\begin{corollary}{\rm (\cite{coal0,coal2})}
\label{cor:bound1}
If $G$ is a graph with $\delta(G) = 1$, then $C(G) \le 2(\Delta(G) + 1)$.
\end{corollary}

\section{Trees with $\mathcal{C}(T)=n$}

In this section, we characterize trees $T$ of order~$n$ satisfying $C(T) = n$. First we characterize graphs $G$ of order $n$ and $\delta(G)=0$ such that $\mathcal{C}(G)=n$.

\begin{theorem}
\label{thmf0}
If $G$ is a graph of order~$n$ with $\delta(G) = 0$, then $\mathcal{C}(G) = n$ if and only if $G \cong  K_1 \cup K_{n-1}$.
\end{theorem}
\begin{proof}
Clearly, $\mathcal{C}(K_1\cup K_{n-1}) = n$. Conversely, suppose that $G$ is a graph of order $n$ with $\delta(G) = 0$ and $\mathcal{C}(G) = n$. Thus, $\{\{v\} \colon v\in V(G)\}$ is a $\mathcal{C}$-partition of $G$. Let $x$ be an isolated vertex of $G$. If $v$ and $v'$ are two distinct vertices in $V(G) \backslash \{x\}$, then since the set $\{v,v'\}$ does not dominate the graph $G$, the pair $\{v\}$  and $\{v'\}$ do not form a coalition. Therefore, $\{v\}$ and $\{x\}$ form a coalition for each $v \in V (G)\backslash\{x\}$, implying that $V(G) \backslash\{x\}$ is a clique in $G$ of cardinality $n - 1$. Thus, $G \cong  K_1 \cup K_{n-1}$.~\QED
\end{proof}

We next define a family $\mathcal{F}_1$ of graphs.

\begin{definition}{\rm (The family $\mathcal{F}_1$)}
\label{defn1}
{\rm Let $G$ be an isolate-free graph constructed as follows. Let $V(G) = \{x,y,w\} \cup P \cup Q$, where $P \cap Q \cap  \{x,y,w\} = \emptyset$ and $|P \cup Q| \ge 1$. Further, if $Q \ne \emptyset$, then $|Q| \ge 2$. Let the edge set $E(G)$ be defined as follows. Let $N_G(x) = \{y\}$ and let $N_G(w) = P \cup Q$. Join each vertex $p \in P$ to every vertex in $(P \cup Q) \setminus \{p\}$. If $Q \ne \emptyset$, then add all edges from the vertex $y$ to every vertex in $Q$, and so $Q \cup \{x\} \subseteq N_G(y)$. Further if $Q \ne \emptyset$, then add edges between vertices in $Q$, including the possibility of adding no edge (in which case $Q$ is an independent set), in such a way that $G[Q]$ does not contain a full vertex. Thus, every vertex $q \in Q$ is not adjacent in $G$ to at least one vertex in $Q \setminus \{q\}$. Finally, add any number of edges between the vertex~$y$ and vertices in $P$, including the possibility of adding no edge between~$y$ and vertices in $P$. }
\end{definition}

We note that if $G \in \mathcal{F}_1$ is a disconnected graph of order~$n$, then $G \cong K_2 \cup K_{n-2}$. We are now in a position to characterize graphs $G$ of order~$n$ with $\delta(G)=1$ and with no full vertex that satisfy $\mathcal{C}(G)=n$.

\begin{theorem}
\label{thmd1}
If $G$ is a graph of order~$n$ with $\delta(G)=1$ and with no full vertex, then $\mathcal{C}(G)=n$ if and only if $G \in {\cal F}_1$.
\end{theorem}
\begin{proof}
Suppose firstly that $G \in {\cal F}_1$. Adopting our notation in Definition~\ref{defn1}, we have that if $p \in P \cup \{w\}$, then $\{x,p\}$ is a dominating set of $G$, and if $q \in Q \cup \{w\}$, then $\{y,q\}$ is a dominating set of $G$, implying that every subset $\{v\}$, where $v \in N[w]$, forms a coalition with $\{x\}$ or $\{y\}$. Therefore, $\mathcal{C}(G)=n$.

Conversely, let $G$ be a graph of order~$n$ with $\delta(G)=1$ and with no full vertex, and suppose that $\mathcal{C}(G)=n$. Let $x$ be a leaf of $G$ and let $y$ be the unique neighbor of $x$ in $G$. Let $\Psi$ be a $\mathcal{C}$-partition of $G$, and so $|\Psi| = n$ and $\Psi$ is a $c$-partition of $G$. Thus, each member of $\Psi$ is a singleton set (of cardinality~$1$), that is, if $v \in V(G)$, then $\{v\} \in \Psi$. We show that $G\in {\cal F}_1$. We know that $\{x\} \in \Psi$ and $\{y\}\in \Psi$. Suppose that $\{x\}$ and $\{y\}$ form a coalition in $G$. Since $x$ is a leaf and $N(x)=\{y\}$, all members  of $V \setminus \{x,y\}$ must be adjacent to $y$. Thus, the degree of~$y$ is $n-1$, and so $y$ is a full vertex of $G$, a contradiction.  Hence, $\{x\}$ and $\{y\}$ do not form a coalition. Let $w$ be a vertex not dominated by~$y$, and so $w \ne y$ and $w$ is not adjacent to~$y$. Let $N(w)=\{w_1,\dots,w_k\}$.

We show that $V = \{x,y\} \cup N[w]$. Suppose, to the contrary, that there exists a vertex $v \in V \backslash (\{x,y\} \cup N[w])$. If $z$ is an arbitrary vertex distinct from $x$ and $y$, then the pair $\{v\}$ and $\{z\}$ do not form a coalition since the vertex $x$ is not dominated by the set $\{v,z\}$. Moreover, $\{v\}$ does not form a coalition with $\{x\}$ or $\{y\}$ since the vertex $w$ is not dominated by the set $\{v,x,y\}$. Hence, $\{v\}$ has no coalition partner, a contradiction. Therefore, $V = \{x,y\} \cup N[w]$. Recall that $w \ne y$ and $w$ is not adjacent to~$y$. Since~$G$ has no full vertices, every set of $\Psi$ must be a coalition partner of some other set of $\Psi$. Let
\[
{\cal W} = \{S\in \Psi \colon N[w] \cap S\ne \emptyset\}.
\]

Since~$\mathcal{C}(G)=n$, we note that ${\cal W} = \left\{\{w\},\{w_1\},\ldots,\{w_k\}\right\}$. Necessarily, each member of $\cal W$ forms a coalition with~$\{x\}$ or~$\{y\}$, and there is no coalition between the members of $\cal W$. In particular, $\{w\}$ is a coalition partner of each of~$\{x\}$ and~$\{y\}$. Let ${\cal W}_P \subseteq {\cal W} \setminus \left\{\{w\}\right\}$ be the collection of all sets  of $\Psi$ that form a coalition with $\{x\}$, and let $P = \{ p \in N(w) \colon \{p\} \in {\cal W}_P \}$. Moreover, let ${\cal W}_Q = \left({\cal W} \setminus \left\{\{w\}\right\}\right) \setminus {\cal W}_P$ and let $Q = \{ q \in N(w) \colon \{q\} \in {\cal W}_Q \}$. We note that $\{P,Q\}$ is a weak partition of $N(w)$, that is, $N(w) = P \cup Q$, where $P \cap Q = \emptyset$ and where in a weak partition we allow some of the sets in the partition to be empty (in our case, possibly $P = \emptyset$ or $Q = \emptyset$). No set in ${\cal W}_Q$ forms a coalition with $\{x\}$, and therefore every member of ${\cal W}_Q$ forms a coalition with $\{y\}$.

Let $p \in P$. Thus, $\{p\} \in {\cal W}_P$ forms a coalition with $\{x\}$, and so the set $\{x,p\}$ is a dominating set of $G$, implying that the vertex $p$ dominates the set $N[w]$, and so $p$ is adjacent to every vertex in $N(w) \backslash \{p\}$. Let $q \in Q$. Thus, $\{q\} \in {\cal W}_Q$ and the set $\{x,q\}$ is not a dominating set of $G$, implying that the vertex~$q$ is not adjacent to at least one vertex in $N(w) \setminus \{q\}$. As observed earlier, every vertex in $P$ is adjacent to every vertex in $N(w) = P \cup Q$, implying that all non-neighbors of $q$ in $N(w)$ belong to the set $Q$. Further, since $\{q,y\}$ is a dominating set of $G$, the vertex~$y$ is adjacent to every vertex in $Q$ that is not adjacent to~$q$. In particular, if $q'$ is a vertex in $Q \setminus \{q\}$ that is not adjacent to $q$, then since $\{y,q'\}$ is a dominating set of $G$, the vertex $y$ is adjacent to the vertex~$q$. Therefore, the vertex $y$ dominates the set $Q$. According to the definition of family~${\cal F}_1$, we infer that $G\in {\cal F}_1$.~\QED
\end{proof}

\medskip
Next we characterize graphs $G$ of order~$n$ with $\delta(G)=1$ and exactly one full vertex that satisfy $\mathcal{C}(G)=n$.


\begin{theorem}
\label{col1}
If $G$ is a graph of order~$n \ge 3$ with $\delta(G)=1$ and with exactly one full vertex, then $\mathcal{C}(G)=n$ if and only if $G$ is obtained from the graph $K_1 \cup K_{n-1}$ by adding an edge joining the isolated vertex to an arbitrary vertex of the complete graph $K_{n-1}$.
\end{theorem}
\begin{proof}
For $n \ge 3$, let $G$ be obtained from $K_1 \cup K_{n-1}$ by adding an edge $xy$ where $x$ is the vertex in the copy of $K_1$ and $y$ is a vertex in the complete graph $K_{n-1}$. In the resulting graph $G$, the vertex $x$ has degree~$1$ and the vertex $y$ has degree~$n-1$. Thus, $y$ is a full vertex. If $w \in V(G) \setminus \{x,y\}$, then $\{w,x\}$ is a dominating set of $G$, and so $\{w\}$ forms a coalition with the set $\{x\}$, implying that $\mathcal{C}(G)=n$. This proves the sufficiency.

To prove the necessity, let $G$ be a graph of order~$n \ge 3$ with $\delta(G)=1$ and with exactly one full vertex and suppose that $\mathcal{C}(G)=n$. Let $\Psi$ be a $\mathcal{C}$-partition of $G$. Let $u$ and $v$ be vertices of $G$ where $\deg_G(u)=1$ and $\deg_G(v)=n-1$, and let $w$ be an arbitrary vertex in $V(G) \setminus \{u,v\}$. Since $v$ is a full vertex, the set $\{v\}$ does not form a coalition with any vertex of $G$. Hence since the set $\{w\}$ is not a dominating set of $G$, the only coalition partner of the set $\{w\}$ in $\Psi$ is the set $\{u\}$. This implies that the vertex $w$ is adjacent to every vertex in $G$ except for the vertex~$u$. This is true for all vertices $w \in V(G) \setminus \{u,v\}$, implying that the graph $G' = G - v \cong K_1 \cup K_{n-2}$. Rebuilding the graph $G$ from $G'$ by adding back the vertex~$y$ and all edges joining $y$ to the $n-1$ vertices in $V(G')$, the desired result follows.~\QED
\end{proof}

We are now in a position to characterize all trees $T$ of order~$n$ with $\mathcal{C}(T)=n$.

\begin{theorem}
\label{thmtree}
If $T$ is a tree of order~$n$, then $\mathcal{C}(T) = n$ if and only if $T$ is a path of order at most~$4$.
\end{theorem}
\begin{proof}
By Theorem~\ref{thm:path}(a), if $T = P_n$ where $n \in [4]$, then $C(P_n) = n$. Conversely, let $T$ be a tree of order~$n$ satisfying $\mathcal{C}(T) = n$. The result is immediate if $n = 1$ or $n = 2$, and so we may assume that $n \ge 3$. If there is a full vertex in $T$, then $T \cong K_{1,n-1}$. In this case, by Theorem~\ref{col1} we must have that $n = 3$, yielding $T = P_3$, as desired. Hence we may assume that $T$ does not have a full vertex, implying that $n \ge 4$. In this case by Theorem~\ref{thmd1}, the tree $T$ belongs to the family~${\cal F}_1$. Adopting our earlier notation in Definition~\ref{defn1}, if $Q \ne \emptyset$, then $|Q| \ge 2$, and the graph $T[\{y,q_1,q_2,w\}]$ contains a $4$-cycle where $\{q_1,q_2\} \subseteq Q$, contradicting the fact that $T$ is a tree. Hence, $Q = \emptyset$. If $|P| \ge 2$, then the graph $T[P \cup \{w\}]$ contains a copy of $K_3$, once again contradicting the fact that $T$ is a tree. Hence, $|P| = 1$, implying that $T$ is the path $P_4$. Therefore, if $\mathcal{C}(T) = n$, then $T = P_n$ where $n \in [4]$.~\QED
\end{proof}

\section{Trees with $\mathcal{C}(T)=n-1$}

In this section, we characterize all trees $T$ with $\mathcal{C}(T) = n-1$. Recall that for $r, s \ge 1$, a double star $S(r,s)$ is a tree with exactly two (adjacent) vertices that are not leaves, with one of the vertices having $r$ leaf neighbors and the other $s$ leaf neighbors. We shall prove the following result.

\begin{theorem}
\label{thm:tree2}
If $T$ is a tree of order~$n$, then $\mathcal{C}(T) = n-1$ if and only if $T \in \{K_{1,3},P_5,P_6,S(2,1)\}$.
\end{theorem}
\begin{proof}
Let $T$ be a tree of order~$n$. If $T \in \{K_{1,3},P_5,P_6,S(2,1)\}$, then it is straightforward to check that $\mathcal{C}(T) = n-1$. To prove the necessity, suppose that $\mathcal{C}(T) = n-1$. By Theorem~\ref{thmtree}, $\mathcal{C}(T) \le n-1$ if and only if $T = K_{1,3}$ or the order of $T$ is at least~$5$. If $T$ contains a full vertex, then $T$ is a star $K_{1,n-1}$ and $\mathcal{C}(T) = 3$. By supposition, $\mathcal{C}(T) = n-1$, implying that in this case, $T = K_{1,3}$ as desired. Hence, we may assume that $T$ does not contain a full vertex, implying that $n \ge 5$.

Let $\Psi$ be a $\mathcal{C}$-partition of $G$, and so $|\Psi| = n-1$ and $\Psi$ is a $c$-partition of $G$. Thus, each member of $\Psi$ is a singleton set of cardinality~$1$, except for one member in $\Psi$ of cardinality~$2$. Let $\Psi = \{V_1, \ldots, V_{n-1},U\}$ where $|V_i| = 1$ for $i \in [n-1]$ and $|U| = 2$. Let $x$ be a leaf of $T$ and let $y$ be the unique neighbor of $x$. We proceed further with the following series of structural properties of the tree $T$.

\begin{claim}
\label{c:claim-1}
$U \ne \{x,y\}$.
\end{claim}
\proof Suppose, to the contrary, that $U = \{x,y\}$. Let $v \in V(T) \setminus U$, and so $\{v\} \in \Psi$. If $\{v'\}$ is a coalition partner of $\{v\}$ for some $v' \in V(T) \setminus U$, then the vertex $x$ is not dominated by the set $\{v,v'\}$, a contradiction. Hence, the coalition partner of $\{v\}$ is the set $U$. In particular, $\{v,y\}$ is a dominating set for all $v \in V(T) \setminus U$. Let $A$ be the set of vertices in $V(T) \setminus U$ that are not adjacent to the vertex~$y$, and let $B$ be the set of vertices in $V(T) \setminus U$ that are dominated by~$y$. If $A = \emptyset$, then $y$ would be a full vertex, a contradiction. Hence, $A \ne \emptyset$. As observed earlier, if $a \in A$, then $\{a\}$ and $\{x,y\}$ form a coalition, implying that the vertex~$a$ is adjacent to every other vertex of $A$. Since this is true for all vertices $a \in A$, we infer that $A$ is a clique. If $B = \emptyset$, then $T$ would be disconnected, a contradiction. Hence, $B \ne \emptyset$. If $b \in B$, then $\{b\}$ and $\{x,y\}$ form a coalition, implying that the vertex~$b$ is adjacent to every vertex of $A$. Since this is true for all vertices $b \in B$, every vertex in $A$ is adjacent to every vertex in $B$. Hence, $\{a,x\}$ is a dominating set for every vertex $a \in A$, and $\{b,y\}$ is a dominating set for every vertex $b \in B$, implying that $\mathcal{C}(T) = n$, a contradiction.~\smallqed

\medskip
By Claim~\ref{c:claim-1}, $U \ne \{x,y\}$, and so $|U \cap \{x,y\}| \le 1$. Since $T$ has no full vertex, no vertex dominates $T$. Let $w$ be a vertex that is not dominated by $y$. Let $N(w) =\{w_1,\ldots,w_k\}$, and so $\deg_G(w) = k$ and $N(w) \cap \{x,y\} = \emptyset$. If two neighbors of $w$ are adjacent, then these two vertices together with $w$ induce a triangle in $T$, a contradiction. Hence since $T$ is a tree, the set $N(w)$ is an independent set. If the vertex $y$ is adjacent to two vertices in $N(w)$, then these two vertices together with $y$ and $w$ induce a $4$-cycle in $T$, a contradiction. Hence the vertex $y$ is adjacent to at most one vertex in $N(w)$.

\begin{claim}
\label{c:claim-2}
If $U \subseteq N(w)$, then $T = P_5$.
\end{claim}
\proof Suppose that $U \subseteq N(w)$, and so $U \cap \{x,y\} = \emptyset$. Renaming vertices if necessary, we may assume that $U = \{w_1,w_2\}$. Suppose that $k \ge 3$. If $U$ forms a coalition with the set $\{x\}$, then $U$ dominates the set $N(w)$, contradicting our earlier observation that the set $N(w)$ is an independent set. Hence, $U$ does not form a coalition with $\{x\}$, implying that $U$ forms a coalition with the set $\{y\}$. Thus, $\{y,w_1,w_2\}$ is a dominating set in $G$. Since $N(w)$ is an independent set, the vertex $y$ therefore dominates the set $N(w) \setminus \{w_1,w_2\}$. Thus, $y$ is adjacent to at least~$k-2$ vertices in $N(w)$. As observed earlier, $y$ is adjacent to at most one vertex in $N(w)$. Therefore, $k = 3$ and $y$ is adjacent to exactly one vertex in~$N(w)$, namely to the vertex~$w_3$. Since $\{x,w_3\}$ does not dominate the set $U$, the pair $\{x\}$ and $\{w_3\}$ do not form a coalition. Since $\{y,w_3\}$ does not dominate the set $U$, the pair $\{y\}$ and $\{w_3\}$ do not form a coalition. However, then the set $\{w_3\} \in \Psi$ does not form a coalition with any other set in $\Psi$, a contradiction. Hence, $k = 2$, that is, $U = N(w) = \{w_1,w_2\}$.

Suppose that $n \ge 6$. Let $z \in V(T) \setminus \{x,y,w,w_1,w_2\}$. By our earlier observations, $\{z\}$ forms a coalition with $\{x\}$ or with $\{y\}$. Thus $\{x,z\}$ or $\{y,z\}$ is a dominating set of $G$. However neither $y$ nor $z$ belong to $N(w)$, and so $\{x,y,z\}$ is not a dominating set of $G$, a contradiction. Hence, $n = 5$, that is, $V(T) = \{x,y,w,w_1,w_2\}$. Since $T$ is connected, the vertex $y$ is adjacent to at least one of $w_1$ and $w_2$. However as observed earlier, the vertex $y$ is adjacent to at most one of $w_1$ and $w_2$. Consequently, the vertex $y$ is adjacent to exactly one of $w_1$ and $w_2$, implying that $T$ is a path $P_5$.~\smallqed

\begin{claim}
\label{c:claim-3}
If $U \subseteq N[w]$, then $T = P_5$.
\end{claim}
\proof Suppose that $U \subseteq N[w]$, and so $U \cap \{x,y\} = \emptyset$. By Claim~\ref{c:claim-2}, if $U \subseteq N(w)$, the $T = P_5$ as desired. Hence renaming vertices if necessary, we may assume that $U = \{w,w_1\}$. Recall that $n \ge 5$. Let $z \in V(T) \setminus \{x,y,w,w_1\}$. Thus, $\{z\}$ forms a coalition with $\{x\}$ or with $\{y\}$. Thus $\{x,z\}$ or $\{y,z\}$ is a dominating set of $G$. Since $y$ is not adjacent with $w$, the vertex $y \notin N[w]$, and therefore the vertex $z$ is necessary adjacent to $w$. Since $z$ is an arbitrary vertex in $V(T) \setminus \{x,y,w,w_1\}$, we have that $V(T) = \{x,y\} \cup N[w]$. We now consider the vertex $w_2$. If $\{w_2\}$ forms a coalition with $\{x\}$, then $\{x,w_2\}$ is a dominating set, implying that $w_1$ and $w_2$ are adjacent, contradicting our earlier observation that $N(w)$ is an independent set. Hence, $\{w_2\}$ forms a coalition with $\{y\}$, then $\{y,w_2\}$ is a dominating set, implying that $y$ and $w_1$ are adjacent. If $k \ge 3$, then the vertex $w_3$ is not dominated by $\{y,w_2\}$, a contradiction. Hence, $k = 2$, and so $T$ is a path $P_5$.~\smallqed

\medskip
By Claim~\ref{c:claim-3}, we may assume that $|U \cap N[w]| \le 1$, for otherwise $T = P_5$ and the desired result follow. By our earlier assumptions, $|U \cap \{x,y\}| \le 1$.

\begin{claim}
\label{c:claim-4}
If $y \in U$, then $T \in \{P_5,S(2,1)\}$.
\end{claim}
\proof Suppose that $y \in U$, and so $x \notin U$. Let $U = \{y,z\}$. By definition of a coalition partition, the set $U$ is not a dominating set of $G$ since the only sets $S \in \Psi$ that dominate $G$ have cardinality $|S| = 1$. Recall that $w$ was chosen earlier as an arbitrary vertex that is not adjacent to~$y$. Renaming vertices if necessary, we can choose $w$ to be a vertex not dominated by the set $U$. Thus, $\{x,y,z\} \cap N[w] = \emptyset$. Recall that $N(w) = \{w_1,\ldots,w_k\}$. By our earlier observations, the set $N(w)$ is an independent set. Further since there is no $4$-cycle in $T$, every vertex in $V(T) \setminus N[w]$ is adjacent to at most one vertex in $N(w)$.

Since the set $\{x,y,z\}$ is not a dominating set of $T$, the sets $\{x\} \in \Psi$ and $U \in \Psi$ do not form a coalition. Hence, the set $\{x\} \in \Psi$ forms a coalition with a set $\{x'\} \in \Psi$, where $x' \notin \{x,y,z\}$. Since $\{x,x'\}$ is a dominating set of $T$, necessarily $x' \in N[w]$. Since the vertex $z$ is not dominated by the set $\{x,w\}$, we note that $x' \ne w$, implying that $x' \in N(w)$. Renaming vertices if necessary, we may assume that $x' = w_1$. Thus, $\{x,w_1\}$ is a dominating set of $G$. Since $N(w)$ is an independent set, this in turn implies that $k = |N(w)| = 1$ and that $w_1$ is adjacent to the vertex~$z$. Suppose that $n \ge 6$, and so there exists a vertex $q \in V(T) \setminus \{x,y,w,w_1,z\}$. Since $\{x,q\}$ is not a dominating set of $T$, the set $\{q\} \in \Psi$ does not form a coalition with $\{x\}$, implying that the set $\{q\} \in \Psi$ forms a coalition with the set $U \in \Psi$. However, the vertex $w$ is not dominated by the set $\{q,y,z\}$, a contradiction. Hence, $n = 5$, and so $V(T) = \{x,y,w,w_1,z\}$. As observed earlier, $\{xy,zw_1,ww_1\} \subset E(T)$. Since $T$ is connected, the only remaining edge of $T$ is either $yw_1$ or $yz$. If $yw_1 \in E(T)$, then $T = S(2,1)$, while if $yz \in E(T)$, then $T = P_5$.~\smallqed

\medskip
By Claim~\ref{c:claim-4}, we may assume that $y \notin U$, for otherwise the desired result follows.

\begin{claim}
\label{c:claim-5}
If $x \notin U$, then $T \in \{P_5,S(2,1)\}$.
\end{claim}
\proof Suppose that $x \notin U$. By assumption, $y \notin U$. Hence, $\{x\} \in \Psi$ and $\{y\} \in \Psi$. Since $T$ has no full vertices, the pair $\{x\}$ and $\{y\}$ do not form a coalition. By our earlier assumptions, $|U \cap N[w]| \le 1$. Let $U = \{z_1,z_2\}$ and let $z_1 \in U \setminus N[w]$. By supposition, $U \cap \{x,y\} = \emptyset$, and so $z_1 \notin \{x,y\}$. Since $U$ forms a coalition with $\{x\}$ or $\{y\}$, the set $\{x,y,z_1,z_2\}$ is a dominating set of $T$, implying that $z_2 \in N[w]$ in order to dominate the vertex~$w$.

Suppose that there is a vertex $z \in V(T) \setminus \left(\{x,y,z_1\} \cup N[w]\right)$. Since $z \notin U$, $\{z\} \in \Psi$. By our earlier observations, the pair $\{z\}$ and $\{x\}$ form a coalition or the pair $\{z\}$ and $\{y\}$ form a coalition. However, the set $\{x,y,z\}$ does not dominate the vertex~$w$, and so $\{z\}$ does not form a coalition with $\{x\}$ or with $\{y\}$, a contradiction. Hence, $V(T) = \{x,y,z_1\} \cup N[w]$.

Suppose that $k \ge 2$. Renaming vertices in $N(w)$ if necessary, we may assume that $w_1 \ne z_2$, that is, $w_1 \notin U$. Thus, $\{w_1\} \in \Psi$.  Since $\{x,w_1\}$ is not a dominating set of $T$ noting that $N(w)$ is an independent set of cardinality~$k \ge 2$, the pair $\{w_1\}$ and $\{x\}$ do not form a coalition, implying that the pair $\{w_1\}$ and $\{y\}$ form a coalition. Thus, the vertex~$y$ dominates the set $N(w) \setminus \{w_1\}$. Since $y$ is adjacent to at most one vertex in $N(w)$, we infer that $k = 2$, $yw_2 \in E(T)$, and $yw_1 \notin E(T)$. Analogously, if $w_2 \ne z_2$, then $yw_1 \in E(T)$ and $yw_2 \notin E(T)$, a contradiction. Hence, $w_2 = z_2$. Thus, $\{w\} \in \Psi$. Since $\{x,w\}$ does not dominate the vertex~$z_1$, the pair $\{w\}$ and $\{x\}$ do not form a coalition, implying that the pair $\{w\}$ and $\{y\}$ form a coalition. This in turn implies that $yz_1 \in E(T)$. The tree is now determined and $E(T) = \{xy,yz_1,yw_2,w_2w,ww_1\}$. Further, $\Psi = \{ \{x\}, \{y\}, \{w\}, \{w_1\}, \{z_1,w_2\} \}$. However, $\{x\}$ does not form a coalition with any other set in $\Psi$, a contradiction.

Hence, $k = 1$, and so $V(T) = \{x,y,w,w_1,z_1\}$ and $\{xy,ww_1\} \subset E(T)$. Since $T$ is a tree, exactly two of the edges in the set $\{yw_1,yz_1,w_1z_1\}$ are present in $T$. If $\{yw_1,yz_1\} \subset E(T)$ or $\{yw_1,w_1z_1\} \subset E(T)$, then $T = S(2,1)$. If $\{yz_1,w_1z_1\} \subset E(T)$, then $T = P_5$.~\smallqed

\medskip
By Claim~\ref{c:claim-5}, we may assume that $x \in U$, for otherwise the desired result follows. Since $x$ is an arbitrary vertex of degree~$1$ in $T$, we may therefore assume that the set $U$ contains all vertices of degree~$1$ in $T$. Since $|U| = 2$, this implies that $T$ is a path $P_n$ where $n \ge 5$. By Theorem~\ref{thm:path}(a), $C(P_n) = n-1$ if $n \in \{5,6\}$ and $C(P_n) \le n-2$ if $n \ge 7$. Hence since $T = P_n$ and $\mathcal{C}(T) = n-1$, we infer that $T \in \{P_5,P_6\}$. This completes the proof of Theorem~\ref{thm:tree2}.~\QED
\end{proof}

\medskip
As a consequence of Theorem~\ref{thm:tree2}, we have the following result.

\begin{corollary}
If $T$ is a tree of order $n$ with $n\ge 7$, then $\mathcal{C}(T) \le n-2$.
\end{corollary}

\section{Coalition graphs}
	
In this section, we solve Problem~\ref{prob2} and Problem~\ref{prob3}. In particular, we determine the number of the  coalition graphs that can be defined by all coalition partitions of the  path $P_k$ for a given positive integer $k$ (see Theorem~\ref{thma}) and then as a consequence of Theorem~\ref{thma}, we see that there is no universal coalition path.
	
A coalition graph corresponding to a coalition partition of a path is called a \CP-\emph{graph} in \cite{coal1}. We say a path $P_k$ defines a $\CP$-graph $G$ if there is a coalition partition $\Psi$ of $P_k$ such that $\CG(P_k,\Psi)$ is isomorphic to $G$. In  \cite{coal1}, Haynes et al. proved the following theorem.

\begin{theorem}{\rm (\cite{coal1})}
\label{lemCP}
A graph $G$ is a $\CP$-graph if and only if $G\in {\cal F}$.
\end{theorem}
	
According to Theorem~\ref{lemCP}, all 18 graphs of ${\cal F}$ are $\CP$-graph. For $k \ge 2$, let the ordered list $s_1, s_2,\ldots, s_k$ be the vertices of the path $P_k$  whose  edges are $\{s_i,s_{i+1}\}$ where $i \in [k-1]$. Now, we prove the  following proposition.

\begin{proposition}
\label{lem33}
The following properties hold. \\ [-24pt]
\begin{enumerate}
\item[{\rm (a)}] The path $P_1$ defines the $\CP$-graph $K_1$.
\item[{\rm (b)}] The path $P_2$ defines the $\CP$-graph $\overline{K_2}$.
\item[{\rm (c)}] The path $P_3$ defines the $\CP$-graph $K_1\cup K_2$.
\item[{\rm (d)}] For any integer $k\ge 4$, the path $P_k$ does not define the $\CP$-graphs $\overline{K_2}$ and $K_1\cup K_2$.
\end{enumerate}
\end{proposition}
\begin{proof} It is immediate that the only path that defines the $\CP$-graph $K_1$ is $P_1$. This establishes~(a). For the path $P_2$, the set $\left\{\{s_1\},\{s_2\}\right\}$ is a coalition partition and the corresponding coalition graph is $\overline{K_2}$. This establishes~(b). For the path $P_3$, the set $\left\{\{s_1\}, \{s_2\}, \{s_3\}\right\}$ is a coalition partition such that none of the sets $\{s_1\}$ and $\{s_3\}$ is a dominating set but they form a coalition. Moreover, the set $\{s_2\}$ is a singleton dominating set. Hence, the  coalition graph corresponding to  $\left\{\{s_1\}, \{s_2\}, \{s_3\}\right\}$ is $K_1\cup K_2$. This establishes~(c). Since the graphs $\overline{K_2}$ and $K_1\cup K_2$ have at least one isolated vertex, their corresponding coalition partition must contain at least one singleton dominating set. We know that for any $k \ge 4$, the domination number of $P_k$ is at least two. Therefore, every coalition partition of $P_k$  with $k\ge 4$ does  not contain any singleton dominating set. This proves part~(d).~\QED
\end{proof}

\medskip
By Proposition~\ref{prop:Delta} and Theorem~\ref{thm:path}(a), we can readily deduce the following result.

\begin{corollary}
\label{cor1}
For any positive integer $k\le 3$, the path $P_k$ does not define the $\CP$-graphs $P_3, K_3, K_{1,3}$,  $2K_2$, $P_4, C_4, F_1, K_4-e, P_2\cup P_3, F_2, B_1, P_5, S(2,1), S(2,2)$.
\end{corollary}

Next we prove several propositions about paths and their corresponding \CP-graphs.
	
\begin{proposition}\label{prop1}
For $k \ge 4$, the path $P_k$ defines the $\CP$-graph $K_2$.
\end{proposition}
\begin{proof}
For $k \ge 4$, let $\Psi = \{A, B\}$, where  $A=\{s_1,s_2,\ldots, s_{\lfloor\frac{k}{2}\rfloor}\}$  and $B=\{s_{\lfloor\frac{k}{2}\rfloor+1},\ldots, s_k\}$. The sets $A$ and $B$ form a coalition, and so $\CG(P_k,\Psi)\cong  K_2$.~\QED
\end{proof}

\begin{proposition}\label{prop2}
For $k \ge 4$, the path $P_k$ defines the $\CP$-graphs $P_3$ and $C_4$.
\end{proposition}
\begin{proof}
To prove the proposition, we  provide two  coalition partitions $\Psi_1$ and $\Psi_2$ for $P_k$  with $k\ge 4$ whose corresponding coalition graphs are  $P_3$ and  $C_4$, respectively. Let $\Psi_1=\{A, B, C\}$, where
\[
A=\{s_1,s_2\}, \hspace*{0.5cm} B=\bigcup_{i=2}^{\lceil\frac{k}{2}\rceil}\{s_{2i-1}\} \hspace*{0.5cm} \mbox{and} \hspace*{0.5cm} C=\bigcup_{i=2}^{\lfloor\frac{k}{2}\rfloor}\{s_{2i}\}.
\]

The sets $A$ and $B$ form a coalition, as do the sets $A$ and $C$. Moreover, the sets $B$ and $C$ do not form a coalition. Hence, $\CG(P_k,\Psi_1)\cong  P_3$. Let $\Psi_2=\{A, B, C, D\}$, where
\[
A=\{s_1\}, \hspace*{0.5cm} B=\{s_2\}, \hspace*{0.5cm} C= \bigcup_{i=0}^{\lfloor\frac{k}{2}\rfloor-2}\{s_{k-(2i+1)}\} \hspace*{0.5cm} \mbox{and} \hspace*{0.5cm} D=\bigcup_{i=0}^{\lceil\frac{k}{2}\rceil-2}\{s_{k-2i}\}.
\]

The set $A$ forms a coalition with both sets $C$ and $D$, while the set $B$ forms a coalition with both sets $C$ and $D$. Moreover, the sets $A$ and $B$ do not form a coalition, and the sets $C$ and $D$ do not form a coalition. Hence, $\CG(P_k,\Psi_2)\cong  C_4$.~\QED
\end{proof}

\begin{proposition}\label{prop3}
For $k \ge 5$, the path $P_k$ defines the $\CP$-graph $F_1$.
\end{proposition}
\begin{proof}
Let $\Psi=\{A, B, C, D\}$ be the coalition partition of $P_k$ defined as follows. For $k=6$, let $A=\{s_1,s_5\}$, $B=\{s_3,s_6\}$, $C=\{s_4\}$ and $D=\{s_2\}$.  For $k \ge 8$ even, let
\[
A=\left(\bigcup_{i=2}^{\frac{k}{2}-2}\{s_{k-2i}\}\right)\cup\{s_1, s_k\}, \hspace*{0.5cm}  B=\left(\bigcup_{i=2}^{\frac{k}{2}-2}\{s_{k-(2i-1)}\}\right)\cup \{s_2\}, \hspace*{0.5cm}
C=\{s_{k-2}\}  \hspace*{0.5cm} \mbox{and} \hspace*{0.5cm}
D=\{s_3, s_{k-1}\}.
\]

For $k \ge 5$ odd, let
\[
A=\left(\bigcup_{i=2}^{\frac{k-3}{2}}\{s_{k-2i}\}\right)\cup\{s_1, s_k\}, \hspace*{0.5cm}  B=\left(\bigcup_{i=2}^{\frac{k-3}{2}}\{s_{k-(2i-1)}\}\right)\cup \{s_2\}, \hspace*{0.5cm}
C=\{s_{k-2}\}  \hspace*{0.5cm} \mbox{and} \hspace*{0.5cm}
D=\{s_{k-1}\}.
\]

The set $A$ forms a coalition with each of the sets $B$, $C$ and $D$, and the sets $B$ and $D$ form a coalition. Moreover, the sets $B$ and $C$ do not form a coalition, and the sets $C$ and $D$ do not form a coalition. Hence, $\CG(P_k,\Psi)\cong F_1$.~\QED
\end{proof}

\begin{proposition}
\label{prop4}
For $k \ge 6$, the path $P_k$ defines the $\CP$-graphs $K_{1,3}$, $K_3$, $K_4-e$, and $P_2 \cup P_3$.
\end{proposition}
\begin{proof}
To prove the proposition, we  provide four  coalition partitions $\Psi_1, \Psi_2, \Psi_3$ and $\Psi_4$ for $P_k$  with $k\ge 6$ whose corresponding coalition graphs are  $K_{1,3}, K_3, K_4-e$ and $P_2\cup P_3$, respectively.

Let $\Psi_1=\{A, B, C, D\}$, where $A=\{s_1,s_2, \ldots, s_{k-4}\}\cup\{s_k\}$, $B=\{s_{k-3}\}$, $C=\{s_{k-2}\}$ and $D=\{s_{k-1}\}$. The set $A$ forms a coalition with each of the sets $B$, $C$ and $D$. Moreover, there are no coalitions among the sets $B$, $C$ and $D$. Hence, $\CG(P_k,\Psi) \cong K_{1,3}$. Let $\Psi_2=\{A, B, C\}$, where

\[
A=\left(\bigcup_{i=2}^{{\lfloor\frac{k}{3}\rfloor-1}}\{s_{k-3i}\}\right)\cup \{s_1, s_{k-2}\}, \hspace*{0.5cm}
B=\left(\bigcup_{i=1}^{\lfloor\frac{k}{3}\rfloor-1}\{s_{k-(3i+1)}\}\right)\cup \{s_2, s_{k}\} \hspace*{0.5cm}
\]
and
\[
C=\left(\bigcup_{i=1}^{{\lfloor\frac{k+1}{3}\rfloor-2}}\{s_{k-(3i+2)}\}\right)\cup \{s_{k-1}, s_{k-3}\}.
\]

The set $A$ forms a coalition with each of the sets $B$ and $C$. Moreover, the sets $B$ and $C$ form a coalition. Hence, $\CG(P_k,\Psi_2)\cong  K_3$. Let $\Psi_3=\{A, B, C, D\}$, where the coalition partition is defined as follows. For $k \ge 6$ even, let $C=\{s_{k-3}\}$ and $D=\{ s_{k-2}\}$, and let
\[
A=\left(\bigcup_{i=0}^{\frac{k}{2}-3}\{s_{2i+1}\}\right)\cup\{s_{k-1}\}
\hspace*{0.25cm} \mbox{and} \hspace*{0.25cm} B=\left(\bigcup_{i=1}^{\frac{k}{2}-2}\{s_{2i}\}\right)\cup \{s_k\}.
\]

For $k \ge 7$ odd, let $C=\{s_{k-3}\}$ and $D=\{ s_{k-2}\}$, and let
\[
A=\left(\bigcup_{i=0}^{\frac{k-5}{2}}\{s_{2i+1}\}\right)\cup\{s_{k}\}
\hspace*{0.25cm} \mbox{and} \hspace*{0.25cm}
B=\left(\bigcup_{i=1}^{\frac{k-5}{2}}\{s_{2i}\}\right)\cup \{s_{k-1}\}.
\]

The set $A$ forms a coalition with each of the sets $B$, $C$ and $D$. Moreover, the set $B$ forms a coalition with each of the sets $C$ and $D$. However, the sets $C$ and $D$ do not form a coalition. Hence, $\CG(P_k,\Psi_3)\cong  K_4-e$. Let $\Psi_{4}=\{A, B, C, D, F\}$, where the coalition partition is defined as follows. For $k \ge 6$ even, let $B=\{s_2, s_{k-4}\}$, $C=\{ s_{k-3}\}$, and $D = \{ s_{k-2}\}$ and let
\[
A=\left(\bigcup_{i=0}^{\frac{k-6}{2}}\{s_{2i+1}\}\right)\cup\{s_{k}\}
\hspace*{0.25cm} \mbox{and} \hspace*{0.25cm}
F=\left(\bigcup_{i=2}^{\frac{k-6}{2}}\{s_{2i}\}\right)\cup \{s_{k-1}\}.
\]

For $k \ge 7$ odd, let $B=\{s_2, s_{k-5}\}$, $C=\{ s_{k-4}\}$, and $D=\{ s_{k-3}\}$,  and let
\[
A=\left(\bigcup_{i=0}^{\frac{k-7}{2}}\{s_{2i+1}\}\right)\cup\{s_{k-1}\}
\hspace*{0.25cm} \mbox{and} \hspace*{0.25cm}
F=\left(\bigcup_{i=2}^{\frac{k-7}{2}}\{s_{2i}\}\right)\cup \{s_{k-2}, s_{k}\}.
\]

The sets $A$ and $C$, the sets $A$ and $D$, and the sets $B$ and $F$ are the only pairs among the sets $A$, $B$, $C$, $D$, and $F$ in $\Psi_{4}$ that form a coalition. Hence, $\CG(P_k,\Psi_4)\cong  P_2\cup P_3$.~\QED
\end{proof}

\begin{proposition}
\label{prop5}
For $k \ge 7$, the path $P_k$ defines the $\CP$-graphs $P_4$ and $P_5$.
\end{proposition}
\begin{proof}
To prove the proposition, we  provide two  coalition partitions $\Psi_1$ and $\Psi_2$ for $P_k$  with $k \ge 7$ whose corresponding coalition graphs are $P_4$ and $P_5$, respectively. Let $\Psi_1=\{A, B, C, D\}$, where $C=\{s_{k-4}\}$ and $D=\{s_{k-1}\}$, and where
\[
A=\left(\bigcup_{i=3}^{\lceil\frac{k}{2}\rceil-2}\{s_{k-2i}\}\right)\cup\{s_1, s_{k-2}, s_k\}
\hspace*{0.25cm} \mbox{and} \hspace*{0.25cm}
B=\left(\bigcup_{i=3}^{\lfloor\frac{k}{2}\rfloor-1}\{s_{k-(2i-1)}\}\right)\cup \{s_2, s_{k-3}\}.
\]

The sets $C$ and $A$, the sets $A$ and $B$, and the sets $B$ and $D$ are the only pairs among the sets $A$, $B$, $C$, and $D$ in $\Psi_{1}$ that form a coalition. Hence, $\CG(P_k,\Psi_1) \cong  P_4$. Let $\Psi_{2}=\{A, B, C, D, F\}$, where $C=\{ s_3\}$, $D=\{ s_4\}$,  and $F = \{s_{5}\}$, and where
\[
A=\left(\bigcup_{i=3}^{\lfloor \frac{k}{2} \rfloor}\{s_{2i}\}\right)\cup\{s_{1}\}
\hspace*{0.25cm} \mbox{and} \hspace*{0.25cm}
B=\left(\bigcup_{i=3}^{\lceil \frac{k}{2} \rceil-1}\{s_{2i+1}\}\right)\cup\{s_{2}\}.
\]

The sets $C$ and $A$, the sets $A$ and $D$, the sets $D$ and $B$, and the set $B$ and $F$  are the only pairs among the sets $A$, $B$, $C$, $D$, and $F$ in $\Psi_{2}$ that form a coalition. Hence, $\CG(P_k,\Psi_2) \cong  P_5$.~\QED
\end{proof}

\begin{proposition}
\label{prop6}
For $k \ge 8$, the path $P_k$ defines the $\CP$-graph $2K_2$.
\end{proposition}
\begin{proof}
Let $\Psi=\{A, B, C, D\}$ be the coalition partition of $P_k$ defined as follows.

For $k = 8$, let  $A = \{s_1,s_7\}$, $B = \{s_2,s_8\}$, $C = \{s_3,s_4\}$ and $D = \{s_5,s_6\}$. 

For $k=9$, let  $A = \{s_1,s_8\}$, $B = \{s_2,s_4,s_9\}$, $C = \{s_3,s_5\}$ and $D = \{s_6,s_7\}$. 

For $k=10$, let  $A = \{s_1,s_9\}$, $B = \{s_2,s_4,s_{10}\}$, $C = \{s_3,s_5,s_6\}$ and $D = \{s_7,s_8\}$. 

For $k=11$, let  $A = \{s_1,s_{10}\}$, $B = \{s_2,s_4,s_{11}\}$, $C = \{s_3,s_5,s_7\}$ and $D = \{s_6, s_8,s_9\}$. 

For $k \ge 12$, let  $A=\{s_1, s_{k-3}, s_{k-1}\}$ and $B=\{s_2, s_4, s_k\}$, and let
\[
C=\left(\bigcup_{i=2}^{\lceil\frac{k}{2}\rceil-3}\{s_{2i+1}\}\right)\cup\{s_3\}
\hspace*{0.25cm} \mbox{and} \hspace*{0.25cm}
D=\left(\bigcup^{\lceil\frac{k}{2}\rceil-2}_{i=3}\{s_{2i}\}\right)\cup\{s_{k-2}\}.
\]

The sets $A$ and $C$, and the sets $B$ and $D$, are the only pairs among the sets $A$, $B$, $C$, and $D$ in $\Psi$ that form a coalition. Hence, $\CG(P_k,\Psi) \cong  2K_2$.~\QED
\end{proof}

\begin{proposition}
\label{prop7}
For $k \ge 9$, the path $P_k$ defines the $\CP$-graphs $F_2$, $B_1$, and $S(2,1)$.
\end{proposition}
\begin{proof}
		To prove the proposition, we  provide three  coalition partitions 	$\Psi_1, \Psi_2$ and $\Psi_3$ for $P_k$  with $k\ge 9$ whose corresponding coalition graphs are $F_2$, $B_1$, and $S(2,1)$, respectively. Let $\Psi_{1}=\{A, B, C, D, F\}$, where $C=\{ s_3,s_6\}$, $D=\{ s_4, s_7\}$,  and $F= \{s_{5}\}$, and where
\[
A=\left(\bigcup_{i=4}^{\lceil \frac{k}{2} \rceil-1}\{s_{2i+1}\}\right)\cup\{s_{1}\}
\hspace*{0.25cm} \mbox{and} \hspace*{0.25cm}
B=\left(\bigcup_{i=4}^{\lfloor \frac{k}{2} \rfloor}\{s_{2i}\}\right)\cup\{s_{2}\}.
\]

The set $A$ forms a coalition with each of the sets $C$ and $D$, and the set $B$ forms a coalition with each of the sets $C$, $D$ and $F$. However, these are the only pairs of sets in $\Psi_{1}$ that form a coalition, implying that $\CG(P_k,\Psi_1) \cong  F_2$. Let  $\Psi_{2}=\{A, B, C, D, F\}$, where  $C=\{ s_3\}$, $D=\{ s_5\}$,  and $F= \{s_{7}\}$, and where
\[
A=\left(\bigcup_{i=4}^{\lfloor \frac{k}{2} \rfloor}\{s_{2i}\}\right)\cup\{s_{1}, s_4\}
\hspace*{0.25cm} \mbox{and} \hspace*{0.25cm}
B=\left(\bigcup_{i=4}^{\lceil \frac{k}{2} \rceil-1}\{s_{2i+1}\}\right)\cup\{s_{2}, s_6\}.
\]

The set $A$ forms a coalition with each of the sets $B$, $D$ and $F$, and the set $B$ forms a coalition with each of the sets $A$, $C$ and $D$. However, these are the only pairs of sets in $\Psi_{2}$ that form a coalition, implying that $\CG(P_k,\Psi_2) \cong  B_1$. Let $\Psi_{3}=\{A, B, C, D, F\}$, where $C=\{ s_3\}$, $D=\{ s_4\}$, and $F = \{s_{7}\}$, and where
\[
A=\left(\bigcup_{i=3}^{\lfloor \frac{k}{2} \rfloor}\{s_{2i}\}\right)\cup\{s_{1}\}
\hspace*{0.25cm} \mbox{and} \hspace*{0.25cm}
B=\left(\bigcup_{i=4}^{\lceil \frac{k}{2} \rceil-1}\{s_{2i+1}\}\right)\cup\{s_{2}, s_5\}.
\]

The set $A$ forms a coalition with each of the sets $B$, $C$ and $D$, and the set $B$ forms a coalition with each of the sets $A$ and $F$. However, these are the only pairs of sets in $\Psi_{2}$ that form a coalition, implying that $\CG(P_k,\Psi_2) \cong  S(2,1)$.~\QED
\end{proof}

\begin{proposition}
\label{prop8}
For $k \ge 10$, the path $P_k$ defines the $\CP$-graph $S_{2,2}$.
\end{proposition}
\begin{proof}
By Theorem~\ref{thm:path}(a), we have $C(P_n) = 6$ for all $k \ge 10$. Therefore, there is a coalition partition of $P_k$ such that the corresponding $\CP$-graph has six vertices.  On the other hand, the only $\CP$-graph with six vertices is $S_{2,2}$. Hence, the path $P_k$ with $k \ge 10$ defines the $\CP$-graph~$S_{2,2}$.~\QED
\end{proof}

\medskip
In Propositions~\ref{prop1}--\ref{prop8}, we presented results that determine which $\CP$-graphs can be defined by the path $P_k$. But there are still cases that have not yet been identified. For example, does $P_8$ define the $\CP$-graphs $F_2$ and $B$  or does $P_6$ define the $\CP$-graphs $2K_2$ and $P_4$?  To this end, we empirically checked the remaining cases.  Table~\ref{T11} summarizes the results of Corollary~\ref{cor1} and  Propositions~\ref{lem33}-\ref{prop8},  and the empirical results. In the table, we used the letters $Y, N, y$ and $n$.  Let $T(i,j)$ be the cell of the table in row $i$ and column $j$. Let $\CP_i$ be  the $\CP$-graph in row $i$,  and let $P_j$ be the path in column $j$. When $T(i,j)$ is equal to a uppercase letter, it means that  using  Corollary~\ref{cor1} and  Propositions~\ref{lem33}-\ref{prop8}, we have obtained the value of $T(i,j)$, and when $T(i,j)$ is equal to a lowercase letter, it means that we empirically have obtained it. If $T(i,j)\in \{Y, y\}$, it means that the $\CP$-graph $\CP_i$ can be defined by the path $P_j$, and if $T(i,j)\in \{N, n\}$, it means that  the $\CP$-graph $\CP_i$ cannot be defined by the path $P_j$.

\begin{table}[hbt]
		\centering
		\begin{tabular}{|l|l|l|l|l|l|l|l|l|l|l|}
			\hline
			& $P_1$ & $P_2$ & $P_3$ & $P_4$ & $P_5$ & $P_6$ & $P_7$ & $P_8$ & $P_9$ & $P_{k\ge 10}$ \\ \hline
			$K_1$&$Y$	  &   $N$    &    $N$   &   $N$    &  $N$     &   $N$    &    $N$   &   $N$    &    $N$   &       $N$         \\ \hline
			$K_2$            &     $N$  &   $n$    &  $Y$     &   $Y$    &    $Y$   & $Y$      &$Y$       &$Y$       &$Y$       &$Y$                \\ \hline
			$\overline{K_2}$ &$N$       &  {$Y$}     &    $n$   & $N$      &$N$       & $N$      &$N$       &   $N$    &  $N$     & $N$               \\ \hline
			$K_1\cup K_2$    & $N$      &     $N$  &  {$Y$}     & $N$      &$N$       & $N$      &$N$       &   $N$    & 	 $N$     & $N$      \\ \hline
			$P_3$            &    $N$   &     $N$  &     {$N$}  &     $Y$  &   $Y$    &  $Y$     & $Y$      & $Y$      &  $Y$     & $Y$               \\ \hline
			$K_3$            &    $N$   &    $N$   &       $N$&    $n$     &       $y$&    $Y$   &   $Y$    & $Y$      & $Y$      & $Y$               \\ \hline
			$K_{1,3}$        &     $N$  &    $N$   &    $N$   &     $n$  &     $n$   &    $Y$   &   $Y$    & $Y$      & $Y$      & $Y$         \\ \hline
			$2K_2$            &      $N$ &    $N$   &     $N$  &      $n$ &      $n$  &       $y$   &      $y$ &   $Y$    &  $Y$     &  $Y$              \\ \hline
			$P_4$            &       $N$&       $N$&       $N$&       $n$&      $y$ &       $y$   &   $Y$    &       $Y$    &  $Y$     &  $Y$                     \\ \hline
			$C_4$            &       $N$&      $N$ &       $N$&   $Y$    &      $Y$ &  $Y$     &  $Y$     &      $Y$ &   $Y$    &       $Y$         \\ \hline
			$F_1$            &       $N$&       $N$&       $N$&      $n$ &      $Y$ &  $Y$     &  $Y$     &      $Y$ &   $Y$    &       $Y$         \\ \hline
			$K_4-e$         &       $N$&       $N$&      $N$ &     $n$  &       $n$ &  $Y$     &  $Y$     &      $Y$ &   $Y$    &       $Y$         \\ \hline
			$P_2\cup P_3$    & $N$      & $N$      &  $N$     &    $n$   &    $n$   &  $Y$     &  $Y$     &      $Y$ &   $Y$    &       $Y$         \\ \hline
			$F_2$            &       $N$&   $N$    &       $N$&     $n$  &      $n$ &      $n$ &    $n$   &       $n$ &    $Y$   &  $Y$              \\ \hline
			$B$              &       $N$&     $N$  &       $N$&   $n$    &      $n$ &      $n$ &   $n$    &       $n$&   $Y$    & $Y$               \\ \hline
			$P_5$            &     $N$  &   $N$    &     $N$  & $n$      &    $n$   &    $n$   &  $Y$     &$Y$       &  $Y$     &     $Y$           \\ \hline
			$S(2,1)$        &  $N$     &    $N$   &  $N$     &     $n$  & $n$      &      $n$ &       $y$&  $y$     &  $Y$     &  $Y$              \\ \hline
			$S_{2,2}$        &   $N$      &     $N$    &   $N$      &     $n$    &   $n$      &    $n$     &       $n$  &    $n$     &    $n$     &  $Y$              \\ \hline
		\end{tabular}	\caption{The $\CP$-graphs that can be defined by any path $P_k$ with $k\ge 1$.}
		\label{T11}
\end{table}
	
In the following, for the cases that $T(i,j)=y$, we present a coalition partition of the path $P_j$  corresponding  to the $\CP$-graph $\CP_i$. \\ [-22pt]
\begin{enumerate}
\item[$\bullet$] The path $P_8$ and  the $\CP$-graph $S(2,1)$:  $\{\{1,4\}, \{2,6,8\},\{3\}, \{5\}, \{7\}\}$.
\item[$\bullet$] The path $P_7$ and  the $\CP$-graph $S(2,1)$: $\{\{1,4\}, \{2,6\},\{3\}, \{5\}, \{7\}\}$. 
\item[$\bullet$] The path $P_7$ and  the $\CP$-graph $2K_2$: $\{\{1,7\}, \{2\},\{3,4\}, \{5,6\}\}$.
\item[$\bullet$] The path $P_6$ and  the $\CP$-graph $P_4$: $\{\{1,4\}, \{2\},\{3,5\}, \{6\}\}$.
\item[$\bullet$] The path $P_6$ and  the $\CP$-graph $2K_2$: $\{\{1,6\}, \{2\},\{3,4\}, \{5\}\}$.
\item[$\bullet$] The path $P_5$ and  the $\CP$-graph $P_4$: $\{\{1\}, \{2\},\{3,4\}, \{5\}\}$.
\item[$\bullet$] The path $P_5$ and  the $\CP$-graph $K_3$: $\{\{1,5\}, \{2\},\{3,4\}\}$.
\end{enumerate}

Let $\NC(P_i)$ be the number of $\CP$-graphs that can be defined by $P_i$. From Table \ref{T11}, we can readily obtain the following result.
	
\begin{theorem}
\label{thma}
It holds that $\NC(P_1)=\NC(P_2)=1$, $\NC(P_3)=2$, $\NC(P_4)=3$, $\NC(P_5)=6$, $\NC(P_6)=10$, $\NC(P_7)=\NC(P_8)=12$, $\NC(P_9)=14$, and $\NC(P_k)=15$ for any integer $k\ge 10$.
\end{theorem}
	
By Theorem \ref{thma}, there is no path $P_k$ that defines all 18 $\CP$-graphs, yielding the following result.
	
\begin{theorem}
There is no universal coalition path.
\end{theorem}

\section{Conclusion}

In this paper, we characterized all graphs $G$ of order~$n$ with $\delta(G)\le 1$ and $\mathcal{C}(G)=n$. Furthermore, we characterized all trees $T$ of order~$n$ with $\mathcal{C}(T)=n$ and all trees $T$ of order~$n$ with $\mathcal{C}(T)=n-1$. On the other hand, we theoretically and empirically determined the number of coalition graphs that can be defined by all coalition partitions of a given path $P_k$.  Furthermore, we showed that there is no universal coalition path. It remains an open problem to characterize all trees $T$ of order~$n$ with $\mathcal{C}(T)=n-k$ for all $k$ where $2 \le k \le n-2$. It would also be interesting to determine whether there is a linear-time algorithm to compute the coalition number of a given tree.

\end{document}